\documentclass[12pt]{article}
\usepackage{indentfirst}
\usepackage{multirow}
\usepackage{cite}
\usepackage{amsthm}
\usepackage{amsmath}
\usepackage{amssymb}
\usepackage{amsopn}
\usepackage{a4wide}
\usepackage[dvipdfm,
            colorlinks=true,
            unicode,
            linkcolor=blue,
            citecolor=blue,
            urlcolor=blue,
            pagebackref=true,
            pdfpagemode=none,
            pdfstartview=FitH
            ]{hyperref}
\usepackage{color}

\DeclareMathOperator{\wds}{WDS}

\begin{document}

\newtheorem{theorem}{Theorem}
\newtheorem{corollary}{Corollary}
\newtheorem{definition}{Definition}
\newtheorem{lemma}{Lemma}

\title{
A Complete Method for Checking Hurwitz Stability of a Polytope of Matrices
\footnote{
Partially supported by a National Key Basic Research Project of China
(2004CB318000) and by National Natural Science Foundation of China
(10571095)
}
}
\date{}

\author{
Junwei Shao, Xiaorong Hou\footnote{The author to whom all correspondence should be sent.}\\
\textit{\footnotesize School of Automation Engineering,
University of Electronic Science and Technology of China, Sichuan, PRC}\\
\textit{\footnotesize E-mail: \href{mailto:junweishao@gmail.com}{junweishao@gmail.com},
\href{mailto:houxr@uestc.edu.cn}{houxr@uestc.edu.cn} }
}

\maketitle

\noindent\textbf{Abstract:} We present a novel method for checking the
Hurwitz stability of a polytope of matrices. First we prove that
the polytope matrix is stable if and only if two homogenous polynomials
are positive on a simplex,
then through a newly proposed method, i.e., the weighted difference substitution method,
the latter can be checked in finite steps. Examples show the efficiency of our method.
\\[2ex]

\section{Introduction}
Given a linear time-invariant system $\dot{x}(t) = A x(t)$,
it is asymptotically stable if the system matrix $A$ is Hurwitz stable,
i.e., all eigenvalues of A have negative real parts.
Sometimes, we need to consider a type of system uncertainty
in which case the family of system matrices forms a polytope,
i.e., $A$ is varying in
\begin{equation} \label{eqn:polytopeofmatrices}
\mathbf{A} = \left\{
\sum\limits_{k=1}^m q_k A_k:  \sum\limits_{k=1}^m q_k = 1, q_k \geq 0 \mbox{ for all k}
\right\},
\end{equation}
where $A_1,\ldots,A_m \in \mathbb{R}^{n \times n}$ are constant matrices.
We say that $\mathbf{A}$ is robustly Hurwitz stable if each matrix in $\mathbf{A}$ is Hurwitz stable.
The stability of a polytope of matrices cannot be derived from
the stability of all its edges \cite{barmish:1},
which is the case of the stability of a polytope of polynomials \cite{bartlett:1}.
In fact, \cite{cobb:1} proved that for a polytope of $n \times n$ matrices,
the stability of all $2n-4$ dimensional faces can guarantee the stability of the polytope,
and the number $2n-4$ is minimal.
But checking the stability of $2n-4$ dimensional faces of a polytope is also a difficult task.
For a matrix polytope with normal vertex matrices,
\cite{wang:1} proved that the stability of vertex matrices
is necessary and sufficient for the stability of the whole polytope.
Meanwhile,
serval sufficient criterions \cite{fang:1,qian:1,stipanovic:1,dzhafarov:1,dzhafarov:2,oliveira:1,leite:1}
are provided to check stability of matrix polytopes.

Based on the newly proposed results on checking positivity of forms (i.e., homogenous polynomials),
we present a method for checking the stability of a polytope of matrices in this paper,
this method is complete, moreover, it only has a power exponential complexity.

\section{Notations}
\begin{itemize}
\item
$\mathbb{R}$: the field of real numbers.

\item
$\mathbb{Z}$: the set of all integers.

\item
$\mathbb{N}$: the set of all nonnegative integers.

\item
$\mathbb{R}^{m \times n}$: the space of $m \times n$ real matrices.

\item
$I_n$: the identity matrix of order $n$.

\item
$\det A$: the determinant of a square matrix $A$.

\item
$\Delta_f$: the Hurwitz matrix of the polynomial
$f(x) = a_n x^n + a_{n-1} x^{n-1} \ldots + a_0$,
it is an $n \times n$ matrix defined as
$$
\Delta_f = \left(
\begin{array}{lllll}
a_{n-1} & a_{n-3} & a_{n-5} & \cdots & 0\\
a_{n} & a_{n-2} & a_{n-4} & \cdots & 0\\
0 & a_{n-1} & a_{n-3} & \cdots & 0\\
0 & a_{n} & a_{n-2} & \cdots & 0\\
\multicolumn{5}{c}{\dotfill}\\
\multicolumn{5}{c}{\dotfill}
\end{array}
\right).
$$
The successive principal minors of $\Delta_f$ are denoted by $\Delta_k, k=1,2,\ldots,n$.

\item
$(k_1 k_2 \ldots k_m)$: a permutation of $\{1,2,\ldots,m\}$, which changes $i$ to $k_i, i=1,\ldots,m$.

\item
$\Theta_m$: the set of all $m!$ permutations of $\{1,2,\ldots,m\}$.

\item
$\deg(f)$: the degree of a polynomial.

\item
$A^T$: the transpose of a matrix or vector $A$.

\item
$S_{m}$: the $m-1$ dimensional simplex in $\mathbb{R}^m$, i.e.,
$$
S_{m} = \{ (x_1,\ldots,x_{m}): \sum\limits_{i=1}^{m} x_i = 1, x_i \geq 0, i=1,\ldots,m \}.
$$

\item
$[x]$: the largest integer not exceeding the number $x$.

\end{itemize}

\section{Main Results} \label{sec:mainresults}
Suppose
$$
A = \sum\limits_{k=1}^m q_k A_k
$$
is a matrix in $\mathbf{A}$,
denote its characteristic polynomial by
\begin{equation} \label{eqn:charpoly}
f_A(s) \triangleq \det(s I_n - A)=s^n + a_{n-1} s^{n-1} + \ldots + a_1 s + a_0.
\end{equation}
It is well known that, for all $1 \leq i \leq n$,
$(-1)^i a_{n-i}$ equals the sum of all principle minors of order $i$ of $A$,
hence $a_{n-i}$ is a form of degree $i$ on $q_1,q_2,\ldots,q_m$.
Denote by $b_{ij}$ the $(i,j)$th entry of the Hurwitz matrix of $f_A(s)$, then
$b_{ij} =0$, or $b_{ij}$ is a form of degree $2j-i$ on $q_1,q_2,\ldots,q_m$.
Since
$$
\Delta_k = \sum\limits_{(j_1j_2 \ldots j_k) \in \Theta_k} \pm b_{1 j_1} b_{2 j_2} \cdots b_{k j_k},
$$
and
$$
b_{1 j_1} b_{2 j_2} \cdots b_{k j_k} = 0,
$$
or
\begin{equation}
\deg(b_{1 j_1} b_{2 j_2} \cdots b_{k j_k})
= \sum\limits_{i=1}^k (2 j_{i}- i)
= \sum\limits_{i=1}^k i
=  \frac{k(k+1)}{2},\nonumber
\end{equation}
we can see that $\Delta_k$ is a form of degree $k(k+1)/2$ on $q_1,q_2,\ldots,q_m$.

\begin{theorem} \label{lemma:condpolynom}
There exists a Hurwitz stable matrix in the matrix polytope $\mathbf{A}$,
then $\mathbf{A}$ is Hurwitz stable if and only if
\begin{equation} \label{eqn:pstvforms}
\Delta_{n-1}(q_1,\ldots,q_m) > 0 \mbox{ and } a_0(q_1,\ldots,q_m) > 0,  \quad (q_1,\ldots,q_m) \in S_m.
\end{equation}
\end{theorem}

\begin{proof}
The necessity is directly from Routh-Hurwitz criterion.
If $\mathbf{A}$ is not Hurwitz stable,
then by continuity there must exist
a matrix $A$ in $\mathbf{A}$ which has eigenvalues lying on the imaginary axis.
Suppose eigenvalues of $A$ are $s_1,\ldots,s_n$.
If some $s_i$ equals zero, then
$$
a_0 = (-1)^n s_1 s_2 \cdots s_n = 0,
$$
which contradicts the hypothesis $a_0>0$.
If some $s_i$ and $s_j$ are a pair of conjugate eigenvalues of $A$ on the imaginary axis,
then from Orlando's formula \cite{gantmacher:1},
$$
\Delta_{n-1} = (-1)^{\frac{n(n-1)}{2}} \prod_{1 \leq i < j \leq n} (s_i + s_j) = 0,
$$
which contradicts the hypothesis $\Delta_{n-1} > 0$.
\end{proof}

Remark:
\cite{qian:1} also proved that a polytope matrix $\mathbf{A}$ is robustly stable
if and only if an associated form $\gamma \det \widetilde{A}$
is positive on a simplex, where $\widetilde{A}$ is the Kronecker sum of the matrix $A$
in $\mathbf{A}$ with itself, and $\gamma$ is the sign of $\det \widetilde{A}$.
$\widetilde{A}$ is an $n^2 \times n^2$ matrix, hence the degree of $\gamma \det \widetilde{A}$
is much larger than $\Delta_{n-1}$ and $a_0$ in \eqref{eqn:pstvforms}.

A newly proposed method,
i.e. the difference substitution method \cite{yang:1,yang:2},
can be used to check positivity of forms efficiently.
\cite{yao:1} proved that a form is positive
if and only if we can get forms with positive coefficients after
finite steps of varied forms of substitutions,
i.e., weighted difference substitutions(WDS).
\cite{hou:1} further gave a bound for the number of steps required,
and pointed out that
the WDS method is complete in checking positivity or nonnegativity of integral forms.
We will introduce this method more detailedly in Section \ref{sec:wdsmethod}.
Based on this method, we have

\begin{theorem} \label{thm:complexity}
Suppose entries of $A_k,k=1,\ldots,m$ in \eqref{eqn:polytopeofmatrices} are all rational,
the magnitudes of coefficients of $\Delta_{n-1}$ and $a_0$ in \eqref{eqn:pstvforms} are bounded by $M$,
then the Hurwitz stability of $\mathbf{A}$ can be checked
by an algorithm with complexity
$$
O\left(
m^{m+1} n^{2 m^2} ( n^{2m} \ln M + n^{2(m+1)} \ln m + 2(m^2+m n^{2m}) \ln n )
\right).
$$

\end{theorem}

\section{Positivity of forms on simplices} \label{sec:wdsmethod}
In this section, we will introduce the WDS
method for checking positivity of forms \cite{hou:1},
and analyze the complexity of checking Hurwitz stability of matrix polytopes through this method.

Suppose $\theta=(k_1 k_2 \ldots k_m) \in \Theta_m$,
let $P_{\theta}=(p_{ij})_{m \times m}$ be the permutation matrix corresponding
$\theta$, that is
$$
p_{ij}=\left\{
\begin{array}{ll}
1, & j=k_i\\
0, & j \neq k_i
\end{array}
\right..
$$
Given $T_m \in \mathbb{R} ^{m \times m}$, where
\begin{equation}
T_m = \left(
\begin{array}{cccc}
1 & \frac{1}{2}& \ldots & \frac{1}{m} \\[2pt]
0 & \frac{1}{2} & \ldots & \frac{1}{m}\\[2pt]
\vdots & \ddots& \ddots& \vdots\\[2pt]
0 & \ldots & 0 & \frac{1}{m}
\end{array}
\right),
\end{equation}
let
$$A_{\theta}=P_{\theta}T_m,$$
and call it the WDS matrix
determined by the permutation $\theta$.
The variable substitution $\mathbf{x} = A_{\theta} \mathbf{y}$ corresponding
$\theta$ is called a WDS,
where $\mathbf{x} = (x_1,x_2,\ldots,x_m)^T, \mathbf{y} = (y_1,y_2,\ldots,y_m)^T$.

Let $f(\mathbf{x}) \in \mathbb{R}[x_1,x_2,\ldots,x_m]$ be a form, we call
\begin{equation}
\wds(f)=\bigcup\limits_{\theta \in \Theta_m} \{ f(A_{\theta} \mathbf{x}) \}
\end{equation}
the WDS set of $f$,
\begin{equation}
\wds^{(k)}(f)=\bigcup\limits_{\theta_k \in \Theta_m} \cdots \bigcup\limits_{\theta_1 \in \Theta_m}
\{ f(A_{\theta_k} \cdots A_{\theta_1} \mathbf{x}) \}
\end{equation}
the $k$th WDS set of $f$ for positive integer $k$, and set $\wds^{(0)}(f)=\{ f \}$.

Let $\alpha=(\alpha_1,\alpha_2,\ldots,\alpha_m) \in \mathbb{N}^{m}$,
$|\alpha|=\alpha_1+\alpha_2+ \cdots +\alpha_m$.
For a form of degree $d$
$$
f(x_1,x_2,\ldots,x_m)=\sum\limits_{|\alpha|=d} c_\alpha x_1^{\alpha_1}x_2^{\alpha_2}\cdots x_m^{\alpha_m},
$$
if all coefficients $c_\alpha$ are nonzero, we say $f$ has complete monomials.

It is obvious that if there exists $k \in \mathbb{N}$, such that forms in $\wds^{(k)}(f)$
all have complete monomials, and their coefficients are all positive,
then $f$ is positive on $S_m$.
In fact, the reverse is also true, and for integral forms,
the upper bound for $k$ can also be estimated.

\begin{theorem}[\cite{hou:1}] \label{thm:poster}
Suppose $f \in \mathbb{Z}[x_1, x_2, \ldots, x_m]$ is a form of degree $d$,
and the magnitudes of its coefficients are bounded by $M$,
then $f$ is positive on $S_m$,
if and only if there exists $k \leq C_p(M,m,d)$,
such that each form in $\wds^{(k)}(f)$ has complete monomials, and its
coefficients are all positive, where
\begin{equation} \label{eqn:cp}
C_p(M,m,d)= \left[ \dfrac{ \ln \left( 2^{d^m} M^{d^m+1} m^{d^{m+1}+d} d^{(m+1)d+m d^m} (d+1)^{(m-1)(m+2)} \right)}
{\ln m - \ln (m-1)} \right] + 2
\end{equation}
\end{theorem}

Remark:
The $C_p(M,m,d)$ in \eqref{eqn:cp} provides a theoretical upper bound
of the number of steps of substitutions required to check positivity of an integral form.
In practice, numbers of steps used are generally much smaller than this bound \cite{yao:1}.

\begin{proof}[\textbf{proof of Theorem \ref{thm:complexity}}]
A form $f(q_1,\ldots,q_m)$ of degree $d$ has at most
$$
{d+m-1 \choose m-1} \leq (d+1)^{m-1}
$$
monomials
thus the number of arithmetic operations of computing $\wds(f)$
is bounded by
\begin{equation}
m! (d+1)^{m-1} (d+1)^{m(m-1)} \leq m^m (d+1)^{m^2}.
\end{equation}
Moreover
$$
C_p(M,m,d) = O\left(m ( d^m \ln M + d^{m+1} \ln m + (m^2+m d^m) \ln d )\right),
$$
and $\deg(\Delta_{n-1}) = n(n-1)/2, \deg(a_0) = n$,
therefore the complexity of our method for checking the robust Hurwitz stability
of a polytope of $n \times n$ matrices is
$$
O\left(
m^{m+1} n^{2 m^2} ( n^{2m} \ln M + n^{2(m+1)} \ln m + 2(m^2+m n^{2m}) \ln n )
\right).
$$
\end{proof}

\section{Examples}
First we will illustrate our method through an example from \cite{barmish:1}.
Suppose $\mathbf{A}$ is a polytope of following matrices
$$
A_1=\left(
\begin {array}{ccc} -1&0&1\\ \noalign{\medskip}0&-1&0
\\ \noalign{\medskip}-1&0& 0.1\end {array}
\right),
$$
$$
A_2=\left(
\begin {array}{ccc} -1&0&0\\ \noalign{\medskip}0&-1&1
\\ \noalign{\medskip}0&-1& 0.1\end {array}
\right),
$$
$$
A_3=\left(
\begin {array}{ccc} -1&0&-1\\ \noalign{\medskip}0&-1&-1
\\ \noalign{\medskip}1&1& 0.1\end {array}
\right),
$$
let $A = q_1 A_1 + q_2 A_2 + q_3 A_3$,
and denote by $f_A(s)$ the characteristic polynomial of $A$.
The second successive principle minor of the Hurwitz matrix of $f_A(s)$ is
\begin{equation}
\begin{split}
\Delta_2 = &
{\frac {63}{25}}\,{q_{{1}}}^{3}+{\frac {99}{25}}\,{q_{{1}}}^{2}q_{{3}}
+{\frac {243}{50}}\,{q_{{3}}}^{2}q_{{1}}+{\frac {144}{25}}\,q_{{1}}{q_
{{2}}}^{2}+{\frac {153}{25}}\,q_{{1}}q_{{2}}q_{{3}} \\
& +{\frac {144}{25}}
\,{q_{{1}}}^{2}q_{{2}}+{\frac {243}{50}}\,q_{{2}}{q_{{3}}}^{2}+{\frac
{63}{25}}\,{q_{{2}}}^{3}+{\frac {99}{25}}\,{q_{{2}}}^{2}q_{{3}}+{
\frac {171}{50}}\,{q_{{3}}}^{3},
\end{split} \nonumber
\end{equation}
which is obviously positive on $S_3$.
The constant term of $f_A(s)$ is
\begin{equation}
\begin{split}
a_0 = &
{\frac {9}{10}}\,{q_{{1}}}^{3}+{\frac {9}{10}}\,{q_{{2}}}^{3}-{\frac {
23}{5}}\,q_{{1}}q_{{2}}q_{{3}}-{\frac {13}{10}}\,{q_{{1}}}^{2}q_{{3}}+
{\frac {7}{10}}\,{q_{{1}}}^{2}q_{{2}} \\
& -3/10\,q_{{2}}{q_{{3}}}^{2}-3/10
\,{q_{{3}}}^{2}q_{{1}}+{\frac {7}{10}}\,q_{{1}}{q_{{2}}}^{2} -{\frac {
13}{10}}\,{q_{{2}}}^{2}q_{{3}}+{\frac {19}{10}}\,{q_{{3}}}^{3}.
\end{split}
\end{equation}
$a_0$ is not positive on $S_3$ since
the following form (with a difference of a positive constant factor)
belongs to $\wds^{(3)}(a_0)$ and its coefficients are all negative:
\begin{equation}
\begin{split}
& -6516\,q_{{1}}q_{{2}}q_{{3}}-1296\,{q_{{1}}}^{3}-891\,{q_{{2}}}^{3}-
3888\,{q_{{1}}}^{2}q_{{3}}-3888\,{q_{{1}}}^{2}q_{{2}}-1568\,q_{{2}}{q_
{{3}}}^{2} \\
& -2828\,{q_{{3}}}^{2}q_{{1}}-3483\,q_{{1}}{q_{{2}}}^{2}-2223
\,{q_{{2}}}^{2}q_{{3}}-236\,{q_{{3}}}^{3},
\end{split} \nonumber
\end{equation}
therefore the polytope $\mathbf{A}$ is
not Hurwitz stable.

Furthermore, we have checked $900$ polytopes of matrices for $n=2,3,4$ and $m=2,3,4$,
i.e., $100$ polytopes for each pair $(n,m)$.
The vertexes of these polytopes are generated following a similar method
as was described in \cite{leite:1}: their entries are real numbers with $4$ significant numbers
and uniformly distributed in the interval $[-1,1]$,
moreover the maximal real parts of their eigenvalues equal $-0.0001$
(if not so, a shift should be performed).
Table \ref{table:time} shows the time used to check the stability of these polytopes
on a computer equipped with Intel Core 2 Duo E4500 CPU at 2.2 GHz and 4.0 GB of RAM memory,
our program have been implemented in the computer algebra system Maple.

\begin{table}[h]
\centering
\caption{Time used to check robust Hurwitz stability of 100 polytopes for each pair $(n,m)$}
\label{table:time}
\begin{tabular}{c|c|c|c}
\hline
$n$   & $m$ &   number of stable / unstable polytopes & total time (in seconds) \\ \hline
\multirow{3}*{2} & 2 & 67 / 33  & 0.125   \\
& 3 & 29 / 71 & 1.578  \\
& 4 & 11 / 89 & 2537.360  \\
\hline
\multirow{3}*{3} & 2 & 58 / 42 & 0.123  \\
& 3 & 28 / 72  & 2.009  \\
& 4 & 9 / 91 & 665.129  \\
\hline
\multirow{3}*{4} & 2 & 54 / 46  & 0.090  \\
& 3 & 21 / 79 & 2.608  \\
& 4 & 4 / 96 & 1894.602  \\
\hline
\end{tabular}
\end{table}

Remark:
Generally, the time gets longer with the increase of $n$ or $m$,
but there are some extreme examples that
very long time may be spent even for small $n$ and $m$.
For example, in our experiment, it takes 1129.656 and 1172.234 seconds respectively
to check the Hurwitz stability of two matrix polytopes generated
for $(n,m)=(2,4)$, this makes the time corresponding to $(n,m)=(2,4)$ much longer
than that corresponding to $(n,m)=(3,4)$ or $(n,m)=(4,4)$ in Table \ref{table:time}.


\begin{thebibliography}{99}
\bibitem{barmish:1}
B. Ross Barmish, M. Fu, S. Saleh.
Stability of a Polytope of Matrices: Counterexamples.
IEEE Transactions On Automatic Control, 33 (6) (1988) 569-572.

\bibitem{bartlett:1}
A. C. Bartlett, C. V. Hollot, Huang Lin.
Root Locations of an Entire Polytope of Polynomials: It Suffices to Check the Edges,
Math. Control Signals Systems 1 (1988) 61-71.

\bibitem{cobb:1}
J. D. Cobb, C. L. Demarco.
The Minimal Dimension of Stable Faces Required to Guarantee Stability of a Matrix Polytope.
IEEE Transactions On Automatic Control, 34 (9) (1989) 990-992.

\bibitem{wang:1}
Q. Wang.
Necessary and Sufficient Conditions for Stability of a Matrix Polytope with Normal Vertex Matrices.
Automatica, 27 (5) (1991) 887-888.

\bibitem{qian:1}
R. X. Qian, C. L. DeMarco.
An Approach to Robust Stability of Matrix Polytopes Through Copositive Homogeneous Polynomials.
IEEE Transactions On Automatic Control, 37 (1992) 848-852.

\bibitem{fang:1}
Y. Fang, K. A. Loparo, X. Feng,
A sufficient condition for stability of a polytope of matrices.
Systems \& Control Letters, 23 (1994) 237-245.

\bibitem{stipanovic:1}
D. M. Stipanovi\'{c}, D. D. \v{S}iljak.
Stability of Polytopic systems via convex M-matrices and parameter-dependent Liapunov functions.
Nonlinear Analysis, 40 (2000) 589-609.

\bibitem{dzhafarov:1}
V. Dzhafarov, T. B\"{u}y\"{u}kk\"{o}ro\v{g}lu.
On the stability of a convex set of matrices.
Linear Algebra and its Applications, 414 (2006) 547-559.

\bibitem{dzhafarov:2}
V. Dzhafarov, T. B\"{u}y\"{u}kk\"{o}ro\v{g}lu.
On nonsingularity of a polytope of matrices.
Linear Algebra and its Applications, 429 (2008) 1174-1183.

\bibitem{oliveira:1}
R. C.L.F. Oliveira, P. L. D. Peres.
Stability of polytopes of matrices via affine parameter-dependent Lyapunov functions: Asymptotically exact LMI conditions.
Linear Algebra and its Applications 405 (2005) 209¨C228.

\bibitem{leite:1}
V. J. S. Leite, P. L. D. Peres,
An Improved LMI Condition for Robust D-Stability of Uncertain Polytopic Systems.
IEEE Transactions On Automatic Control, 48 (3) (2003) 500-504.

\bibitem{gantmacher:1}
F. R. Gantmacher.
The Theory of Matrices.
Chelsea Publishing Company, 1960.

\bibitem{yang:1}
Lu Yang,
Difference substitution and automated inequality proving.
Journal of Guangzhou University: Natural Science Edition,
2006, 5(2), 1-7 (in Chinese).

\bibitem{yang:2}
Lu Yang.
Solving Harder Problems with Lesser Mathematics.
Proceedings of the 10th Asian Technology Conference in Mathematics,
ATCM Inc, 2005, 37-46.

\bibitem{yao:1}
Yong Yao.
Termination of the Sequence of SDS Sets and Machine
Decision for Positive Semi-definite Forms.
arXiv: 0904.4030.

\bibitem{hou:1}
X. Hou, J. Shao.
Completeness of the WDS method in Checking Positivity of Integral Forms.
arXiv:0912.1649v1.

\end{thebibliography}
\end{document}